\documentclass[prl,twocolumn,amsmath,amssymb,superscriptaddress,graphics,graphicx]{revtex4}

\usepackage[colorlinks,linkcolor=red,citecolor=blue,urlcolor=blue]{hyperref}
\usepackage{amsthm,amsbsy}
\usepackage[usenames]{color}
\usepackage{mathptmx}
\usepackage{bm}
\usepackage{flafter}
\usepackage{graphicx}
\usepackage{paralist}

\newcommand{\ket}[1]{\left|#1\right\rangle}
\newcommand{\bra}[1]{\left\langle #1 \right|}

\newcommand{\op}[1]{\bm{#1}}

\newcommand{\Srel}[2]{S\left(#1 \lVert #2\right)} 
\newcommand{\Srelg}[2]{S\bigl(#1 \lVert #2\bigr)} 

\definecolor{DarkRed}{rgb}{0.7,0,0}

 \newtheorem{theorem}{Theorem}

\DeclareMathOperator{\Tr}{Tr}

\newcounter{ComntCntr} 

\def\M{\op{\Omega}}

\renewcommand{\theenumi}{(\roman{enumi})}
\renewcommand{\labelenumi}{\theenumi}

\begin{document}

\title{Quantifying the Nonclassicality of Operations}

\author{Sebastian Meznaric}
\affiliation{Clarendon Laboratory, Department of Physics, University of Oxford, Oxford, OX1 3PU, United Kingdom}

\author{Stephen R. Clark}
\affiliation{Centre for Quantum Technologies, National University of Singapore, 3 Science Drive 2, 117543 Singapore, Singapore}
\affiliation{Clarendon Laboratory, Department of Physics, University of Oxford, Oxford, OX1 3PU, United Kingdom}

\author{Animesh Datta}
\affiliation{Clarendon Laboratory, Department of Physics, University of Oxford, Oxford, OX1 3PU, United Kingdom}

\date{\today}

\begin{abstract}
Deep insight can be gained into the nature of nonclassical correlations by studying the quantum operations that create them. Motivated by this we propose a measure of nonclassicality of a quantum operation utilizing the relative entropy to quantify its commutativity with the completely dephasing operation. We show that our measure of nonclassicality is a sum of two independent contributions, the \emph{generating power} -- its ability to produce nonclassical states out of classical ones, and the \emph{distinguishing power} -- its usefulness to a classical observer for distinguishing between classical and nonclassical states. Each of these effects can be exploited individually in quantum protocols. We further show that our measure leads to an interpretation of quantum discord as the difference in superdense coding capacities between a quantum state and the best classical state when both are produced at a source that makes a classical error during transmission.
\end{abstract}

\maketitle

\emph{Introduction.} Identifying the resources that underlie quantum advantages in quantum communication and information processing is a crucial question of fundamental and technological importance. Generally, quantum entanglement is ascribed this role due to its necessity in a number of tasks exhibiting quantum advantages~\cite{Horodecki09, Plenio07}. However, quantum enhancements are possible in certain computations with limited amounts of entanglement or even none at all when the involved quantum state is mixed~\cite{Datta07,Datta08,Lanyon08,Passante11,Datta11}. Universal quantum computation with pure states also appears to be possible with little entanglement~\cite{VandenNest12}. In addition to computational advantages, quantum communication can also exhibit advantages over classical communication in the absence of entanglement~\cite{DiVincenzo04,Datta09-1,Modi10}.

Recently, it has been suggested that correlations beyond quantum entanglement might provide an explanation behind quantum enhancements.  One of the most common quantities is the quantum discord~\cite{Ollivier01,Henderson01,Piani11, Modi11}. Quantum discord has recently been interpreted as the difference in the performance of the quantum state merging protocol between a state and its locally decohered equivalent~\cite{Madhok11}, and secondly as quantifying the amount of entanglement consumption in the quantum state merging protocol~\cite{Cavalcanti11}. The role of quantum discord in a more general family of protocols has also been studied~\cite{Madhok12}.

An important difference between quantum discord and entanglement is that the latter is non-increasing, on average, under local operations and classical communication. This is the underlying principle of the resource theory of quantum entanglement. On the other hand, local operations can actually increase quantum discord~\cite{Streltsov11,Gessner12,Ciccarello12, Hu12-1}. Discordant states can be created out of classical states by a local channel if and only if the channel changes the local algebraic structure~\cite{Hu12}, and several authors have studied the evolution of quantum discord under various forms of dynamics~\cite{Datta09,Mazzola10,Auccaise11,Rao11,Shi11}. However, the principles underlying the creation of nonclassical correlations from quantum operations are still lacking.

\begin{figure}[b]
\centering
\includegraphics[scale=0.9]{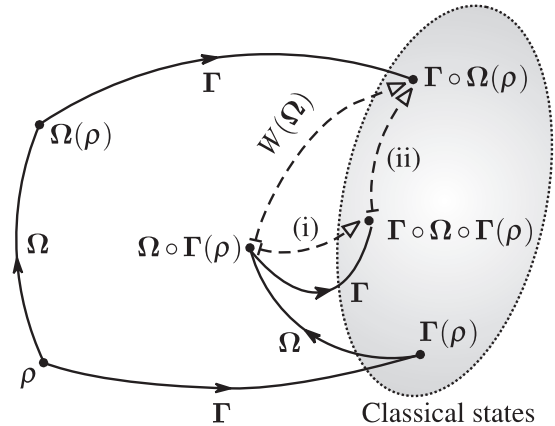}
\caption{Illustration of Thm.~(\ref{th:RelativeEntropySum}). The solid lines represent operations, while the dashed lines represent relative entropies corresponding to the terms from Thm.~\ref{th:RelativeEntropySum}. Since the two curved dashed paths from the state $\M\circ\op{\Gamma}(\rho)$ to $\op{\Gamma}\circ\M(\rho)$ are equidistant in relative entropy the triangle inequality does not apply. Instead quantities (i) and (ii) represent the generating and the distinguishing power, respectively. We consider classical states to be the fixed points of the linear einselection operator $\op{\Gamma}$ and as such is a simplex. This set is smaller than the set of separable~\cite{Horodecki09, Plenio07} and zero-discord states~\cite{Datta11}. Note that our notion of classicality is stricter than that enforced by quantum discord since there is no freedom to choose the classical basis.}
\label{fig:MeasureIllustration}
\end{figure}

Here we investigate the nonclassicality of quantum operations directly. Before presenting our results it is essential to clarify our notion of what is classical. Our criterion is based on einselection, or environment induced superselection, a process via which states of a quantum system become entangled with the environment, effectively \emph{measuring} certain observables of the quantum system~\cite{Zurek03}. We will denote this completely dephasing process as $\op{\Gamma}$. By classicalizing the input and output of a general operation $\M$, a classical operation can be formed as $\op{\Theta} = \op{\Gamma}\circ \M \circ \op{\Gamma}$, where $\circ$ is the composition of operations. Since $\op{\Gamma}^2 = \op{\Gamma}$, this implies the commutation relation $\op{\Theta} \circ \op{\Gamma} = \op{\Gamma} \circ \op{\Theta}$. Taking this relation as the foundation of our notion of classicality will be justified by its implications. We consider classical states $\rho_c$ to be the fixed points of the einselection operator $\op{\Gamma}$ so that $\op{\Gamma}(\rho_c) = \rho_c.$ Thus they are of the form $\rho_c = \sum_{\alpha, \beta} p_{\alpha, \beta} \ket{\alpha}\bra{\alpha} \otimes \ket{\beta}\bra{\beta}$, where $\ket{\alpha}, \ket{\beta}$ are the complete orthonormal eigenbasis of the einselection operator which acted on both parties of the bipartite system. The operation $\op{\Gamma}$ may also act only on one subsystem, in which case we get one-sided classicality with invariant states of the form $\rho = \sum_\alpha p_\alpha \rho_\alpha \otimes \ket{\alpha}\bra{\alpha}$. While all classical states have zero discord, not all zero discord states are classical in the sense used here. A classical observer is one who can measure only in the einselected basis, and for whom the quantum state $\rho$ is completely indistinguishable from the state $\op{\Gamma}(\rho)$. Such states have identical diagonal elements in the einselected basis, but differ in the off-diagonal elements. In contrast a quantum observer may be able to distinguish between the states $\rho$ and $\op{\Gamma}(\rho),$  given enough copies. 

In this Letter, we introduce a measure of quantumness of operations that applies to all completely positive maps. Our approach does not rely on measures of nonclassicality for states and is instead defined from first principles using the fact that a classical map commutes with the einselection operation $\op{\Gamma}$. The extent of non-commutativity is measured using the quantum relative entropy between two different orderings of the dephasing operation $\op{\Gamma}$ and an operation $\M$. We show this measure to be composed of two independent contributions -- firstly, the ability of a nonclassical map to produce non-classical states from classical ones, and secondly, the degree to which it enables classical observers to distinguish states they could not otherwise distinguish classically. We highlight how these two contributions play key roles in quantum communication protocols. Our measure possesses several intuitive properties such as being non-increasing under composition with classical maps and being convex. We calculate our measure for entangling and correlating operations and then apply it to interpret quantum discord via the capacity of superdense coding with noisy states.

We define the \emph{quantumness of an operation} $\M$ as
\begin{equation}
	W_{\Gamma}(\M) = \sup_\rho\, \Srel{\M\circ\op{\Gamma}(\rho)}{\op{\Gamma}\circ\M(\rho)} \label{eq:QDefinition},
\end{equation}
where $\Srel{\rho}{\sigma} = \Tr\left[\rho (\log(\rho) - \log(\sigma))\right]$ is the quantum relative entropy~\cite{Vedral02}, and all logarithms in this paper are base 2. The supremum in Eq.~\eqref{eq:QDefinition} is taken over all quantum states, but we will show shortly it is sufficient to maximize over pure states only. The quantity $W_{\Gamma}$ thus intuitively measures the deviation of the commutator $[\M, \op{\Gamma}]$ from zero by applying both orderings to the same state and comparing the outputs~\footnote{The maximum discrepancy of the outputs over all input states is a measure of distinguishability of operations. This is also the relative entropy of operations and can be used to define conventional properties of channels such as entropy, mutual information, conditional entropy, by applying well understood information theory methods used on quantum states \cite{NielsenChuang}.}. Since the relative entropy is not symmetric in its arguments the specific ordering used in Eq.~\eqref{eq:QDefinition} is essential. The choice of classicalizing the second argument is central to all our subsequent results and is consistent with similar uses of relative entropy in measures of entanglement and discord \cite{Vedral02, Modi10}. The definition of $W_{\Gamma}$ implicitly depends on the fixed einselected basis through $\op{\Gamma}$, however we will suppress this in our subsequent discussions. We now present our main result.\\ \\ \\ \\ \\
\begin{theorem}\label{th:RelativeEntropySum}
The quantumness of an operation $W(\M)$ is the sum of two independent contributions
\begin{align}
	W(\M) & = \sup_\rho \biggl( \Srelg{\op{\Gamma} \circ \M \circ \op{\Gamma} (\rho) }{\op{\Gamma}\circ \M(\rho)} \nonumber \\
	& \quad\quad\quad\quad + \Srelg{\M \circ \op{\Gamma}(\rho)}{\op{\Gamma} \circ \M \circ \op{\Gamma}(\rho)} \biggr), \label{eq:Contributions}
\end{align}
where the supremum is over all quantum states $\rho$.
\end{theorem}
The first term, which we call the \emph{distinguishing power}, characterizes how well a classical observer can distinguish between $\M \circ \op{\Gamma} (\rho)$ and $\M(\rho)$. The second term, which we call the \emph{generating power}, measures the ability of the map $\M$ to generate a nonclassical state out of a classical input. This is depicted in Fig.~(\ref{fig:MeasureIllustration}). Following Klein's inequality, $W(\M)$ vanishes if and only if the operation $\M$ obeys $[\M, \op{\Gamma}]=0$ and so is classical. An implication of Thm.~(\ref{th:RelativeEntropySum}) is that an operation $\M$ is classical only if it has neither distinguishing nor generating power. The proof of Thm.~(\ref{th:RelativeEntropySum}) follows from the monotonicity of relative entropy under completely positive maps \cite{Ruskai02} and is provided in the Supplementary Material.

Crucially, the distinguishing and generating powers in Eq.~\eqref{eq:Contributions} can be independently zero. Given a quantum operation $\op{\Sigma}$ where both these quantities are non-vanishing, we can construct an operation $\op{\Gamma}\circ\op{\Sigma}$ for which the second term in Eq.~\eqref{eq:Contributions} vanishes but the first term is unchanged. Thus, $W(\op{\Gamma}\circ\op{\Sigma})$ is the maximum distinguishing power of $\op{\Sigma}$. On the other hand, for the operation $\op{\Sigma} \circ \op{\Gamma},$ the first term of Eq.~\eqref{eq:Contributions} vanishes while the second one is unchanged and therefore $W(\op{\Sigma}\circ \op{\Gamma})$ is the maximum generating power of $\op{\Sigma}$. By definition, both these quantities are zero for a classical operation $\op{\Theta}$.

There are instances where both terms play essential and independent roles in a quantum protocol. As an example, consider the BB84 quantum cryptography \cite{NielsenChuang}. In order to engage in the protocol, Alice must be able to prepare states in two non-orthogonal bases, which requires only the power to create non-classical states, implying non-vanishing generating power. Bob, on the other hand, needs to be able to distinguish between classical and non-classical states in order to extract the key and detect the presence of an eavesdropper, thus requiring an operation with non-zero distinguishing power.

\renewcommand{\theenumi}{(P\arabic{enumi})}
\renewcommand{\labelenumi}{\theenumi}

The measure of quantumness $W(\M)$ has some additional properties which are physically intuitive, such as

\begin{inparaenum}
\item{ Extremality: Maximum in the supremum is attained with a pure state. This similarly follows from the joint convexity of relative entropy.  \label{prop:PureStateMax}}

\item{ Monotonicity: Given a general operation $\M$ and a classical operation $\op{\Theta}$, then $W(\op{\Theta} \circ \M) \leq W(\M)$ and $W(\M \circ \op{\Theta}) \leq W(\M)$ holds, showing that the measure is non-increasing under composition. \label{prop:ClassicalComposition}}

\item{ Convexity: The convexity follows from the joint convexity of relative entropy. Thus, given two observers with classical maps $\op{\Theta}_i^A, \op{\Theta}_i^B$ at their disposal, and shared source of randomness, they cannot create a nonclassical operation. In other words, if $W(\op{\Theta}_i^A \otimes \op{1}) = 0$ and $W(\op{1} \otimes \op{\Theta}_i^B)=0$ then $W(\sum_i p_i \op{\Theta}_i^A \otimes \op{\Theta}_i^B)=0$. \label{prop:LOSHRMaps}}
\end{inparaenum}

The proofs of the above properties are given in the Supplementary Material as Thm.~2-4. We next evaluate our measure for common decoherence channels, a local discord-generating and a nonlocal entanglement-generating operation. 

\emph{Examples.} Here we focus on qubits with a classical basis as $\ket{0}, \ket{1}$ and $\op{\Gamma}$ implementing two-sided einselection. For unitary operations $\op{U}$ we have that $W(\op{U}) = 0$ if and only if $\op{U}$ is a combination of a classical permutation matrix of the classical basis states with phase shifts, otherwise $W(\op{U}) = \infty$ owing to the logarithm in the definition of relative entropy. This is proved in the Supplementary Material as Thm.~5. For a Hadamard gate $\op{H}$ an infinite quantumness is attained for input states $\ket{\pm} = (\ket{0}\pm \ket{1})/\sqrt{2}$, for which the generating power vanishes and the distinguishing power is infinite. This therefore tells us that the Hadamard gate can be used to ascertain with certainty that an input state is a classical mixture $\rho = (\ket{0}\bra{0} + \ket{1}\bra{1})/2$, and not $\ket{+}$, in a finite number of measurements on average~\cite{Vedral02}.

For standard qubit error models \cite{NielsenChuang}, we similarly find that if the errors occur in the classical basis then they have vanishing $W$. Since Pauli matrices are permutations of the classical basis up to a phase, such models include any Pauli channels on a single qubit, such as the depolarising, bit-flip, phase-flip and the phase-damping channel. The measure $W$ also vanishes for the amplitude-damping channel, $\op{\Xi}_\gamma(\rho) = \op{F}_1 \rho \op{F}_1^\dagger + \op{F}_2 \rho \op{F}_2^\dagger$, where $\op{F}_1 = \ket{0}\bra{0} + \sqrt{1-\gamma} \ket{1}\bra{1}$ and $\op{F}_2 = \sqrt{\gamma}\ket{0}\bra{1}$, since its Kraus operators correspond to permutation matrices, up to a phase. However, if we rotate away from the classical basis, for example by sandwiching an operation between Hadamard gates $\op{H}$, then quantumness may arise. In the Supplementary Material we show that $W(\op{H} \circ \op{\Xi}_\gamma \circ \op{H})$ is non-zero, despite $\op{\Xi}_\gamma$ itself being classical and its removal (or $\gamma = 0$) leaving $\op{H} \circ \op{1} \circ \op{H} = \op{1}$ which is also classical. This illustrates that in general quantum interference makes $W$ non-additive under composition. 

The non-classicality of two unitary operations $\op{U}_1$ and $\op{U}_2$ with infinite $W$ can nonetheless be compared through the use of regularization. For example, $\lim_{\mu \rightarrow 1} W(\op{\Lambda}_\mu \circ \op{U}_1)/W(\op{\Lambda}_\mu \circ \op{U}_2)$ can be evaluated, where $\op{\Lambda}_\mu(\rho) = \mu \rho + (1-\mu)\op{1}/d$ is the depolarising channel for a $d$ dimensional Hilbert space. As such the depolarising channel acts as a regulator and the correct ratio is obtained in the limit where the regulator becomes the identity. This gives a physically motivated ratio of the quantumness of any two operations whenever $W$ diverges. 

An example of a local channel generating non-classical correlations is given in Refs.~\cite{Streltsov11,Campbell11}. The map is of the form $\M = \op{1} \otimes \M_B$, where $\M_B(\rho) = \op{E}_1 \rho \op{E}_1^\dagger + \op{E}_2 \rho \op{E}_2^\dagger$ and $\op{E}_1 = \ket{0}\bra{0}$, $\op{E}_2 = \ket{+}\bra{1}$, which conditionally and irreversibly drives $\ket{1}$ into a state non-orthogonal to $\ket{0}$~\cite{Streltsov11,Campbell11}. Applying the map $\M$ to the classical state $\sigma_c = \frac{1}{2}(\ket{0}\bra{0} \otimes \ket{0}\bra{0} + \ket{1}\bra{1} \otimes \ket{1}\bra{1})$ leads to $\rho = \frac{1}{2}(\ket{0}\bra{0} \otimes \ket{0}\bra{0} + \ket{1}\bra{1} \otimes \ket{+}\bra{+})$, which has non-zero discord~\cite{Shi11b}, but vanishing entanglement since it is a convex combination of product states. Given that $\M_B(\rho)$ possesses diagonal elements which are independent of the off-diagonal elements of $\rho$ it has zero distinguishing power for any input state. The quantumness of $\M$ thus arises from generating power only and is found (see Supplementary Material) to be $W(\M) = 1$, confirming the map is indeed nonclassical.

Next we look at an entangling operation, specifically a CNOT controlled in the $\ket{\pm}$ basis which is capable of generating a maximally entangled two-qubit state from a pure classical input state $\ket{0}$ or $\ket{1}$. As shown in Fig.~\ref{fig:DepolarizedCNOT} we find that when this operation is followed by a joint two-qubit depolarising channel $\op{\Lambda}_\mu$, for $\mu < 2/3$ quantumness is maximized by the generating power alone, while for $\mu > 2/3$ it is maximized purely by the distinguishing power. Thus, at this cross-over point the maximum composing $W$ for this noisy CNOT operation switches from being exposed by its ability to generate nonclassicality to its ability to distinguish nonclassicality.

\begin{figure}[bt]
  \includegraphics[scale=1]{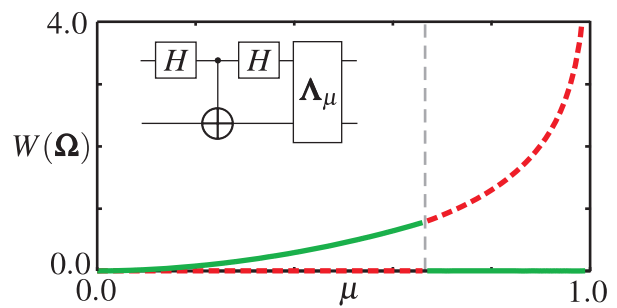}
  \caption{Quantumness of CNOT controlled in the $\ket{\pm}$ basis, followed by the depolarizing channel $\op{\Lambda}_\mu$, as a function of $\mu$. We maximized $W$ and split the expression into the generating power (solid line), and the distinguishing power (dashed line). When $\mu$ is small the action of the depolarizing channel is to substantially degrade distinguishability to such an extent that the generating power dominates. When the $\mu \rightarrow 1$, on the other hand, generating power is fundamentally bounded by $\log(d)$ and thus can no longer compete with the distinguishing power which experiences unbounded growth. The maximum changes from the generating to the distinguishing power at $\mu = 2/3$.}
  \label{fig:DepolarizedCNOT}
\end{figure}

\emph{An interpretation of quantum discord.} Suppose Alice and Bob would like to perform superdense coding, a well-known protocol used to increase the encoding capacity of a single qubit by exploiting entanglement \cite{NielsenChuang}. To do this they order from a source either a quantum state $\rho$, or a cheaper completely dephased version $\op{\Gamma}(\rho)$, to use in the protocol. However, they in fact receive states $\M(\rho)$ or $\M\circ\op{\Gamma}(\rho)$, respectively, where $\M$ accounts for fixed imperfections in the transmission. The question we now ask is how much additional information can they transfer using the superdense coding protocol if they ordered the quantum state $\rho$ rather than $\op{\Gamma}(\rho)$. We will show that if $\M$ is classical, so $W(\M) = 0,$ then the capacity difference is precisely equal to the quantum discord of $\M(\rho),$ where $W$ is evaluated with $\op{\Gamma} = \op{1} \otimes \op{\Gamma}_B$ acting on the receiver's (Bob's) side. This one-sided einselection operator is used to match the definition of the standard quantum discord \cite{Ollivier01,Henderson01}. While this result holds in general, for simplicity we assume that $\rho = \ket{\Phi_d}\bra{\Phi_d}$, where $\ket{\Phi_d} = \sum_\alpha \ket{\alpha}\otimes\ket{\alpha}/\sqrt{d}$ is the maximally entangled state. In this case $\op{\Gamma}(\rho)$ is the maximally classically correlated state.

The capacity of superdense coding~\cite{Bruss04} using a state $\rho$ is given by $F(\rho^{A|B}) = \log(d_A) - S(\rho^{A|B})$, where $S(\rho^{A|B})$ denotes the conditional entropy of the state $\rho$~\footnote{We do not define a conditional state $\rho^{A|B}$, this notation is meant only to denote the conditional entropy of $\rho$.}. Zurek's original definition of quantum discord~\cite{Zurek00} is $Q_z(\rho^{A|B}) = \sum_\alpha p_\alpha S(\rho_A^\alpha) - S(\rho^{A|B})$, where $\rho_A^\alpha$ is the marginal state on Alice's side given that outcome $\alpha$ was obtained, corresponding to the rank-1 projector $\op{\Pi}_\alpha$. Using basic properties of the von Neumann entropy~\cite{NielsenChuang}, we have that $Q_z(\rho^{A|B}) = S\bigl(\op{\Gamma}(\rho^{A|B})\bigr) - S(\rho^{A|B})$. Assuming $[\op{\Gamma}, \M] = 0$, then gives
\begin{multline}
	Q_z(\M\ket{\Phi_d}\bra{\Phi_d}^{A|B}) = \\
	F(\M\ket{\Phi_d}\bra{\Phi_d}^{A|B}) - F(\M \circ \op{\Gamma}\ket{\Phi_d}\bra{\Phi_d}^{A|B}). \label{eq:DiscordInterpretation}
\end{multline}
Extending this to the usual definition of quantum discord $Q$ \cite{Ollivier01, Henderson01}, which involves a minimization over $\op{\Pi}_\alpha$, Eq.~\eqref{eq:DiscordInterpretation} transforms into $Q(\M\ket{\Phi_d}\bra{\Phi_d}^{A|B}) = F(\M\ket{\Phi_d}\bra{\Phi_d}^{A|B})- \sup_\Gamma F(\M\circ\op{\Gamma}\ket{\Phi_d}\bra{\Phi_d}^{A|B})$. Thus quantum discord is the difference in the capacity of superdense coding using the maximally entangled state and the best possible classically correlated state. Our results show that quantum advantage can be gained over the initially classical state in the presence of noise even when $\M(\ket{\Phi_d}\bra{\Phi_d})$ is unentangled. This is illustrated in Fig.~(\ref{fig:SDCapacities}) where $\M = \op{\Lambda}_\mu$ is the depolarising channel.

\begin{figure}[t]
\centering
\includegraphics[scale=1]{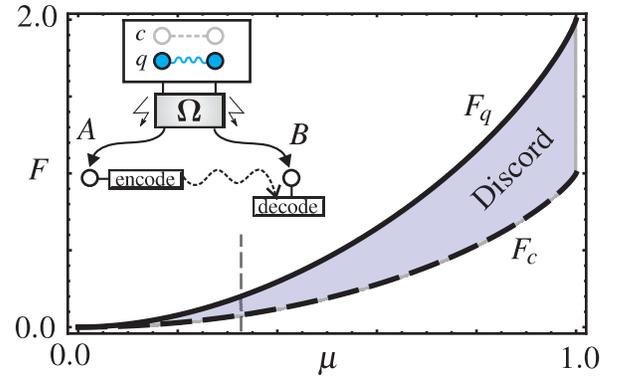}
\caption{Superdense coding capacities $F_q$ using the maximally entangled state and $F_c$ using the classical maximally correlated state when both pass through a depolarising channel $\M = \op{\Lambda}_\mu$. The performances $F_q$ and $F_c$ correspond to the first and the second terms of the Eq. \eqref{eq:DiscordInterpretation}. Quantum discord is then the difference between $F_q$ and $F_c$. The vertical dashed line indicates the value of $\mu$ where all entanglement is lost due to the depolarising channel.}
\label{fig:SDCapacities}
\end{figure}

\emph{Discussion.} In this Letter we have proposed a measure of nonclassicality of quantum operations. The measure is a sum of two independent contributions, the generating power and the distinguishing power, which characterizes an operation as non-classical if and only if the operation can be used by classical observers to distinguish between quantum and classical states or creates nonclassical states out of classical states.

Our measure satisfies several intuitive properties such as convexity and monotonicity under composition of classical maps. In addition, our results show that the einselected relative entropy of discord $Q_g(\rho) = \Srelg{\rho}{\op{\Gamma}(\rho)},$ is non-increasing under the action of classical maps. This is seen by observing that for a classical operation $\op{\Theta},$ we have $Q_g\bigl(\op{\Theta}(\rho)\bigr) = \Srel{\op{\Theta}(\rho)}{\op{\Gamma}\circ \op{\Theta}(\rho)} = \Srel{\op{\Theta}(\rho)}{\op{\Theta} \circ \op{\Gamma} (\rho)} \leq \Srel{\rho}{\op{\Gamma} (\rho)} = Q_g(\rho)$ by monotonicity of relative entropy.  

Furthermore, it is interesting to note that there is a natural complementarity between quantumness of operations and quantumness of states. Specifically, if we denote $\M_\rho$ as any operation capable of generating a state $\rho$ from a classical input then, from Thm.~(\ref{th:RelativeEntropySum}), we see that quantumness must satisfy $W(\M_\rho)\geq Q_g(\rho)$. This provides a readily computable lower bound for $W$. Moreover, through $W$, it suggests a measure of the quantumness of states as $Q_W(\rho) = \inf_{\M_\rho} W(\M_\rho)$, where the minimization is over all operations $\M_\rho$. We conjecture that $Q_W(\rho) = Q_g(\rho)$. If true, this provides a remarkable connection between our measure and quantum discord as well as deepening the link between the nonclassicality of operations and states.

SM would like to thank EPSRC for financial support. SRC thanks the National Research Foundation and the Ministry of Education of Singapore for support. AD was supported in part by the EPSRC (Grant Nos. EP/H03031X/1 and EPSRC/RDF/BtG/0612b/31), and the EU Integrated Project QESSENCE. 

\bibliography{Entanglement}

\begin{thebibliography}{36}
\expandafter\ifx\csname natexlab\endcsname\relax\def\natexlab#1{#1}\fi
\expandafter\ifx\csname bibnamefont\endcsname\relax
  \def\bibnamefont#1{#1}\fi
\expandafter\ifx\csname bibfnamefont\endcsname\relax
  \def\bibfnamefont#1{#1}\fi
\expandafter\ifx\csname citenamefont\endcsname\relax
  \def\citenamefont#1{#1}\fi
\expandafter\ifx\csname url\endcsname\relax
  \def\url#1{\texttt{#1}}\fi
\expandafter\ifx\csname urlprefix\endcsname\relax\def\urlprefix{URL }\fi
\providecommand{\bibinfo}[2]{#2}
\providecommand{\eprint}[2][]{\url{#2}}

\bibitem[{\citenamefont{Horodecki et~al.}(2009)\citenamefont{Horodecki,
  Horodecki, Horodecki, and Horodecki}}]{Horodecki09}
\bibinfo{author}{\bibfnamefont{R.}~\bibnamefont{Horodecki}},
  \bibinfo{author}{\bibfnamefont{P.}~\bibnamefont{Horodecki}},
  \bibinfo{author}{\bibfnamefont{M.}~\bibnamefont{Horodecki}},
  \bibnamefont{and}
  \bibinfo{author}{\bibfnamefont{K.}~\bibnamefont{Horodecki}},
  \bibinfo{journal}{{Reviews of Modern Physics}} \textbf{\bibinfo{volume}{81}},
  \bibinfo{pages}{865} (\bibinfo{year}{2009}).

\bibitem[{\citenamefont{Plenio and Virmani}(2007)}]{Plenio07}
\bibinfo{author}{\bibfnamefont{M.~B.} \bibnamefont{Plenio}} \bibnamefont{and}
  \bibinfo{author}{\bibfnamefont{S.}~\bibnamefont{Virmani}},
  \bibinfo{journal}{{Quantum Information and Computation}}
  \textbf{\bibinfo{volume}{7}}, \bibinfo{pages}{1} (\bibinfo{year}{2007}).

\bibitem[{\citenamefont{Datta and Vidal}(2007)}]{Datta07}
\bibinfo{author}{\bibfnamefont{A.}~\bibnamefont{Datta}} \bibnamefont{and}
  \bibinfo{author}{\bibfnamefont{G.}~\bibnamefont{Vidal}},
  \bibinfo{journal}{{Physical Review A}} \textbf{\bibinfo{volume}{75}},
  \bibinfo{pages}{042310} (\bibinfo{year}{2007}).

\bibitem[{\citenamefont{Datta et~al.}(2008)\citenamefont{Datta, Shaji, and
  Caves}}]{Datta08}
\bibinfo{author}{\bibfnamefont{A.}~\bibnamefont{Datta}},
  \bibinfo{author}{\bibfnamefont{A.}~\bibnamefont{Shaji}}, \bibnamefont{and}
  \bibinfo{author}{\bibfnamefont{C.~M.} \bibnamefont{Caves}},
  \bibinfo{journal}{{Physical Review Letters}} \textbf{\bibinfo{volume}{100}},
  \bibinfo{pages}{050502} (\bibinfo{year}{2008}).

\bibitem[{\citenamefont{Lanyon et~al.}(2008)\citenamefont{Lanyon, Barbieri,
  Almeida, and White}}]{Lanyon08}
\bibinfo{author}{\bibfnamefont{B.~P.} \bibnamefont{Lanyon}},
  \bibinfo{author}{\bibfnamefont{M.}~\bibnamefont{Barbieri}},
  \bibinfo{author}{\bibfnamefont{M.~P.} \bibnamefont{Almeida}},
  \bibnamefont{and} \bibinfo{author}{\bibfnamefont{A.~G.} \bibnamefont{White}},
  \bibinfo{journal}{{Physical Review Letters}} \textbf{\bibinfo{volume}{101}},
  \bibinfo{pages}{200501} (\bibinfo{year}{2008}).

\bibitem[{\citenamefont{Passante et~al.}(2011)\citenamefont{Passante, Moussa,
  Trottier, and Laflamme}}]{Passante11}
\bibinfo{author}{\bibfnamefont{G.}~\bibnamefont{Passante}},
  \bibinfo{author}{\bibfnamefont{O.}~\bibnamefont{Moussa}},
  \bibinfo{author}{\bibfnamefont{D.}~\bibnamefont{Trottier}}, \bibnamefont{and}
  \bibinfo{author}{\bibfnamefont{R.}~\bibnamefont{Laflamme}},
  \bibinfo{journal}{{Physical Review A}} \textbf{\bibinfo{volume}{84}},
  \bibinfo{pages}{044302} (\bibinfo{year}{2011}).

\bibitem[{\citenamefont{Datta and Shaji}(2011)}]{Datta11}
\bibinfo{author}{\bibfnamefont{A.}~\bibnamefont{Datta}} \bibnamefont{and}
  \bibinfo{author}{\bibfnamefont{A.}~\bibnamefont{Shaji}},
  \bibinfo{journal}{{International Journal of Quantum Information}}
  \textbf{\bibinfo{volume}{09}}, \bibinfo{pages}{1787} (\bibinfo{year}{2011}).

\bibitem[{\citenamefont{Van~den Nest}(2012)}]{VandenNest12}
\bibinfo{author}{\bibfnamefont{M.}~\bibnamefont{Van~den Nest}},
  \bibinfo{journal}{{arxiv:1204.3107}}  (\bibinfo{year}{2012}).

\bibitem[{\citenamefont{DiVincenzo et~al.}(2004)\citenamefont{DiVincenzo,
  Horodecki, Leung, Smolin, and Terhal}}]{DiVincenzo04}
\bibinfo{author}{\bibfnamefont{D.}~\bibnamefont{DiVincenzo}},
  \bibinfo{author}{\bibfnamefont{M.}~\bibnamefont{Horodecki}},
  \bibinfo{author}{\bibfnamefont{D.}~\bibnamefont{Leung}},
  \bibinfo{author}{\bibfnamefont{J.}~\bibnamefont{Smolin}}, \bibnamefont{and}
  \bibinfo{author}{\bibfnamefont{B.}~\bibnamefont{Terhal}},
  \bibinfo{journal}{{Physical Review Letters}} \textbf{\bibinfo{volume}{92}},
  \bibinfo{pages}{067902} (\bibinfo{year}{2004}).

\bibitem[{\citenamefont{Datta and Gharibian}(2009)}]{Datta09-1}
\bibinfo{author}{\bibfnamefont{A.}~\bibnamefont{Datta}} \bibnamefont{and}
  \bibinfo{author}{\bibfnamefont{S.}~\bibnamefont{Gharibian}},
  \bibinfo{journal}{{Physical Review A}} \textbf{\bibinfo{volume}{79}},
  \bibinfo{pages}{042325} (\bibinfo{year}{2009}).

\bibitem[{\citenamefont{Modi et~al.}(2010)\citenamefont{Modi, Paterek, Son,
  Vedral, and Williamson}}]{Modi10}
\bibinfo{author}{\bibfnamefont{K.}~\bibnamefont{Modi}},
  \bibinfo{author}{\bibfnamefont{T.}~\bibnamefont{Paterek}},
  \bibinfo{author}{\bibfnamefont{W.}~\bibnamefont{Son}},
  \bibinfo{author}{\bibfnamefont{V.}~\bibnamefont{Vedral}}, \bibnamefont{and}
  \bibinfo{author}{\bibfnamefont{M.}~\bibnamefont{Williamson}},
  \bibinfo{journal}{{Physical Review Letters}} \textbf{\bibinfo{volume}{104}},
  \bibinfo{pages}{080501} (\bibinfo{year}{2010}).

\bibitem[{\citenamefont{Ollivier and Zurek}(2001)}]{Ollivier01}
\bibinfo{author}{\bibfnamefont{H.}~\bibnamefont{Ollivier}} \bibnamefont{and}
  \bibinfo{author}{\bibfnamefont{W.~H.} \bibnamefont{Zurek}},
  \bibinfo{journal}{{Physical Review Letters}} \textbf{\bibinfo{volume}{88}},
  \bibinfo{pages}{017901} (\bibinfo{year}{2001}).

\bibitem[{\citenamefont{Henderson and Vedral}(2001)}]{Henderson01}
\bibinfo{author}{\bibfnamefont{L.}~\bibnamefont{Henderson}} \bibnamefont{and}
  \bibinfo{author}{\bibfnamefont{V.}~\bibnamefont{Vedral}},
  \bibinfo{journal}{{Journal of Physics A: Mathematical and General}}
  \textbf{\bibinfo{volume}{34}}, \bibinfo{pages}{6899} (\bibinfo{year}{2001}).

\bibitem[{\citenamefont{Piani et~al.}(2011)\citenamefont{Piani, Gharibian,
  Adesso, Calsamiglia, Horodecki, and Winter}}]{Piani11}
\bibinfo{author}{\bibfnamefont{M.}~\bibnamefont{Piani}},
  \bibinfo{author}{\bibfnamefont{S.}~\bibnamefont{Gharibian}},
  \bibinfo{author}{\bibfnamefont{G.}~\bibnamefont{Adesso}},
  \bibinfo{author}{\bibfnamefont{J.}~\bibnamefont{Calsamiglia}},
  \bibinfo{author}{\bibfnamefont{P.}~\bibnamefont{Horodecki}},
  \bibnamefont{and} \bibinfo{author}{\bibfnamefont{A.}~\bibnamefont{Winter}},
  \bibinfo{journal}{{Physical Review Letters}} \textbf{\bibinfo{volume}{106}},
  \bibinfo{pages}{220403} (\bibinfo{year}{2011}).

\bibitem[{\citenamefont{Modi et~al.}(2011)\citenamefont{Modi, Brodutch, Cable,
  Paterek, and Vedral}}]{Modi11}
\bibinfo{author}{\bibfnamefont{K.}~\bibnamefont{Modi}},
  \bibinfo{author}{\bibfnamefont{A.}~\bibnamefont{Brodutch}},
  \bibinfo{author}{\bibfnamefont{H.}~\bibnamefont{Cable}},
  \bibinfo{author}{\bibfnamefont{T.}~\bibnamefont{Paterek}}, \bibnamefont{and}
  \bibinfo{author}{\bibfnamefont{V.}~\bibnamefont{Vedral}},
  \bibinfo{journal}{{arxiv:1112.6238}}  (\bibinfo{year}{2011}).

\bibitem[{\citenamefont{Madhok and Datta}(2011)}]{Madhok11}
\bibinfo{author}{\bibfnamefont{V.}~\bibnamefont{Madhok}} \bibnamefont{and}
  \bibinfo{author}{\bibfnamefont{A.}~\bibnamefont{Datta}},
  \bibinfo{journal}{{Physical Review A}} \textbf{\bibinfo{volume}{83}},
  \bibinfo{pages}{032323} (\bibinfo{year}{2011}).

\bibitem[{\citenamefont{Cavalcanti et~al.}(2011)\citenamefont{Cavalcanti,
  Aolita, Boixo, Modi, Piani, and Winter}}]{Cavalcanti11}
\bibinfo{author}{\bibfnamefont{D.}~\bibnamefont{Cavalcanti}},
  \bibinfo{author}{\bibfnamefont{L.}~\bibnamefont{Aolita}},
  \bibinfo{author}{\bibfnamefont{S.}~\bibnamefont{Boixo}},
  \bibinfo{author}{\bibfnamefont{K.}~\bibnamefont{Modi}},
  \bibinfo{author}{\bibfnamefont{M.}~\bibnamefont{Piani}}, \bibnamefont{and}
  \bibinfo{author}{\bibfnamefont{A.}~\bibnamefont{Winter}},
  \bibinfo{journal}{{Phys. Rev. A}} \textbf{\bibinfo{volume}{83}},
  \bibinfo{pages}{032324} (\bibinfo{year}{2011}).

\bibitem[{\citenamefont{Madhok and Datta}(2012)}]{Madhok12}
\bibinfo{author}{\bibfnamefont{V.}~\bibnamefont{Madhok}} \bibnamefont{and}
  \bibinfo{author}{\bibfnamefont{A.}~\bibnamefont{Datta}},
  \bibinfo{journal}{{International Journal of Modern Physics B}}
  \textbf{\bibinfo{volume}{27}}, \bibinfo{pages}{1245041}
  (\bibinfo{year}{2012}).

\bibitem[{\citenamefont{Streltsov et~al.}(2011)\citenamefont{Streltsov,
  Kampermann, and Bru{\ss{}}}}]{Streltsov11}
\bibinfo{author}{\bibfnamefont{A.}~\bibnamefont{Streltsov}},
  \bibinfo{author}{\bibfnamefont{H.}~\bibnamefont{Kampermann}},
  \bibnamefont{and}
  \bibinfo{author}{\bibfnamefont{D.}~\bibnamefont{Bru{\ss{}}}},
  \bibinfo{journal}{{Physical Review Letters}} \textbf{\bibinfo{volume}{107}},
  \bibinfo{pages}{170502} (\bibinfo{year}{2011}).

\bibitem[{\citenamefont{Gessner et~al.}(2012)\citenamefont{Gessner, Laine,
  Breuer, and Piilo}}]{Gessner12}
\bibinfo{author}{\bibfnamefont{M.}~\bibnamefont{Gessner}},
  \bibinfo{author}{\bibfnamefont{E.-M.} \bibnamefont{Laine}},
  \bibinfo{author}{\bibfnamefont{H.-P.} \bibnamefont{Breuer}},
  \bibnamefont{and} \bibinfo{author}{\bibfnamefont{J.}~\bibnamefont{Piilo}},
  \bibinfo{journal}{{Physical Review A}} \textbf{\bibinfo{volume}{85}},
  \bibinfo{pages}{052122} (\bibinfo{year}{2012}).

\bibitem[{\citenamefont{Ciccarello and Giovannetti}(2012)}]{Ciccarello12}
\bibinfo{author}{\bibfnamefont{F.}~\bibnamefont{Ciccarello}} \bibnamefont{and}
  \bibinfo{author}{\bibfnamefont{V.}~\bibnamefont{Giovannetti}},
  \bibinfo{journal}{{Physical Review A}} \textbf{\bibinfo{volume}{85}},
  \bibinfo{pages}{010102} (\bibinfo{year}{2012}).

\bibitem[{\citenamefont{Hu et~al.}(2012{\natexlab{a}})\citenamefont{Hu, Fan,
  Zhoul, and Liu}}]{Hu12-1}
\bibinfo{author}{\bibfnamefont{X.}~\bibnamefont{Hu}},
  \bibinfo{author}{\bibfnamefont{H.}~\bibnamefont{Fan}},
  \bibinfo{author}{\bibfnamefont{D.}~\bibnamefont{Zhoul}}, \bibnamefont{and}
  \bibinfo{author}{\bibfnamefont{W.-M.} \bibnamefont{Liu}},
  \bibinfo{journal}{{arxiv: 1203.6149}}  (\bibinfo{year}{2012}{\natexlab{a}}).

\bibitem[{\citenamefont{Hu et~al.}(2012{\natexlab{b}})\citenamefont{Hu, Fan,
  Zhou, and Liu}}]{Hu12}
\bibinfo{author}{\bibfnamefont{X.}~\bibnamefont{Hu}},
  \bibinfo{author}{\bibfnamefont{H.}~\bibnamefont{Fan}},
  \bibinfo{author}{\bibfnamefont{D.}~\bibnamefont{Zhou}}, \bibnamefont{and}
  \bibinfo{author}{\bibfnamefont{W.-M.} \bibnamefont{Liu}},
  \bibinfo{journal}{{Physical Review A}} \textbf{\bibinfo{volume}{85}},
  \bibinfo{pages}{032102} (\bibinfo{year}{2012}{\natexlab{b}}).

\bibitem[{\citenamefont{Datta}(2009)}]{Datta09}
\bibinfo{author}{\bibfnamefont{A.}~\bibnamefont{Datta}},
  \bibinfo{journal}{{Physical Review A}} \textbf{\bibinfo{volume}{80}},
  \bibinfo{pages}{052304} (\bibinfo{year}{2009}).

\bibitem[{\citenamefont{Mazzola et~al.}(2010)\citenamefont{Mazzola, Piilo, and
  Maniscalco}}]{Mazzola10}
\bibinfo{author}{\bibfnamefont{L.}~\bibnamefont{Mazzola}},
  \bibinfo{author}{\bibfnamefont{J.}~\bibnamefont{Piilo}}, \bibnamefont{and}
  \bibinfo{author}{\bibfnamefont{S.}~\bibnamefont{Maniscalco}},
  \bibinfo{journal}{{Physical Review Letters}} \textbf{\bibinfo{volume}{104}},
  \bibinfo{pages}{200401} (\bibinfo{year}{2010}).

\bibitem[{\citenamefont{Auccaise et~al.}(2011)\citenamefont{Auccaise,
  C{\'{e}}leri, Soares-Pinto, deAzevedo, Maziero, Souza, Bonagamba, Sarthour,
  Oliveira, and Serra}}]{Auccaise11}
\bibinfo{author}{\bibfnamefont{R.}~\bibnamefont{Auccaise}},
  \bibinfo{author}{\bibfnamefont{L.}~\bibnamefont{C{\'{e}}leri}},
  \bibinfo{author}{\bibfnamefont{D.}~\bibnamefont{Soares-Pinto}},
  \bibinfo{author}{\bibfnamefont{E.}~\bibnamefont{deAzevedo}},
  \bibinfo{author}{\bibfnamefont{J.}~\bibnamefont{Maziero}},
  \bibinfo{author}{\bibfnamefont{A.}~\bibnamefont{Souza}},
  \bibinfo{author}{\bibfnamefont{T.}~\bibnamefont{Bonagamba}},
  \bibinfo{author}{\bibfnamefont{R.}~\bibnamefont{Sarthour}},
  \bibinfo{author}{\bibfnamefont{I.}~\bibnamefont{Oliveira}}, \bibnamefont{and}
  \bibinfo{author}{\bibfnamefont{R.}~\bibnamefont{Serra}},
  \bibinfo{journal}{{Physical Review Letters}} \textbf{\bibinfo{volume}{107}},
  \bibinfo{pages}{140403} (\bibinfo{year}{2011}).

\bibitem[{\citenamefont{Rao et~al.}(2011)\citenamefont{Rao, Srikanth,
  Chandrashekar, and Banerjee}}]{Rao11}
\bibinfo{author}{\bibfnamefont{B.}~\bibnamefont{Rao}},
  \bibinfo{author}{\bibfnamefont{R.}~\bibnamefont{Srikanth}},
  \bibinfo{author}{\bibfnamefont{C.}~\bibnamefont{Chandrashekar}},
  \bibnamefont{and} \bibinfo{author}{\bibfnamefont{S.}~\bibnamefont{Banerjee}},
  \bibinfo{journal}{{Physical Review A}} \textbf{\bibinfo{volume}{83}},
  \bibinfo{pages}{064302} (\bibinfo{year}{2011}).

\bibitem[{\citenamefont{Shi et~al.}(2011{\natexlab{a}})\citenamefont{Shi,
  Jiang, Sun, and Du}}]{Shi11}
\bibinfo{author}{\bibfnamefont{M.}~\bibnamefont{Shi}},
  \bibinfo{author}{\bibfnamefont{F.}~\bibnamefont{Jiang}},
  \bibinfo{author}{\bibfnamefont{C.}~\bibnamefont{Sun}}, \bibnamefont{and}
  \bibinfo{author}{\bibfnamefont{J.}~\bibnamefont{Du}}, \bibinfo{journal}{{New
  Journal of Physics}} \textbf{\bibinfo{volume}{13}}, \bibinfo{pages}{073016}
  (\bibinfo{year}{2011}{\natexlab{a}}).

\bibitem[{\citenamefont{Zurek}(2003)}]{Zurek03}
\bibinfo{author}{\bibfnamefont{W.~H.} \bibnamefont{Zurek}},
  \bibinfo{journal}{{Reviews of Modern Physics}} \textbf{\bibinfo{volume}{75}},
  \bibinfo{pages}{715} (\bibinfo{year}{2003}).

\bibitem[{\citenamefont{Vedral}(2002)}]{Vedral02}
\bibinfo{author}{\bibfnamefont{V.}~\bibnamefont{Vedral}},
  \bibinfo{journal}{{Reviews of Modern Physics}} \textbf{\bibinfo{volume}{74}},
  \bibinfo{pages}{197} (\bibinfo{year}{2002}).

\bibitem[{\citenamefont{Ruskai}(2002)}]{Ruskai02}
\bibinfo{author}{\bibfnamefont{M.~B.} \bibnamefont{Ruskai}},
  \bibinfo{journal}{{Journal of Mathematical Physics}}
  \textbf{\bibinfo{volume}{43}}, \bibinfo{pages}{4358} (\bibinfo{year}{2002}).

\bibitem[{\citenamefont{Nielsen and Chuang}(2000)}]{NielsenChuang}
\bibinfo{author}{\bibfnamefont{M.~A.} \bibnamefont{Nielsen}} \bibnamefont{and}
  \bibinfo{author}{\bibfnamefont{I.~L.} \bibnamefont{Chuang}},
  \emph{\bibinfo{title}{{Quantum Computation and Quantum Information}}}
  (\bibinfo{publisher}{{Cambridge University Press}}, \bibinfo{year}{2000}).

\bibitem[{\citenamefont{Campbell et~al.}(2011)\citenamefont{Campbell, Apollaro,
  Di~Franco, Banchi, Cuccoli, Vaia, Plastina, and Paternostro}}]{Campbell11}
\bibinfo{author}{\bibfnamefont{S.}~\bibnamefont{Campbell}},
  \bibinfo{author}{\bibfnamefont{T.~J.~G.} \bibnamefont{Apollaro}},
  \bibinfo{author}{\bibfnamefont{C.}~\bibnamefont{Di~Franco}},
  \bibinfo{author}{\bibfnamefont{L.}~\bibnamefont{Banchi}},
  \bibinfo{author}{\bibfnamefont{A.}~\bibnamefont{Cuccoli}},
  \bibinfo{author}{\bibfnamefont{R.}~\bibnamefont{Vaia}},
  \bibinfo{author}{\bibfnamefont{F.}~\bibnamefont{Plastina}}, \bibnamefont{and}
  \bibinfo{author}{\bibfnamefont{M.}~\bibnamefont{Paternostro}},
  \bibinfo{journal}{{Phys. Rev. A}} \textbf{\bibinfo{volume}{84}},
  \bibinfo{pages}{052316} (\bibinfo{year}{2011}).

\bibitem[{\citenamefont{Shi et~al.}(2011{\natexlab{b}})\citenamefont{Shi, Wang,
  Jiang, and Du}}]{Shi11b}
\bibinfo{author}{\bibfnamefont{M.}~\bibnamefont{Shi}},
  \bibinfo{author}{\bibfnamefont{W.}~\bibnamefont{Wang}},
  \bibinfo{author}{\bibfnamefont{F.}~\bibnamefont{Jiang}}, \bibnamefont{and}
  \bibinfo{author}{\bibfnamefont{J.}~\bibnamefont{Du}}, \bibinfo{journal}{J.
  Phys. A: Math. Theor.} \textbf{\bibinfo{volume}{44}}, \bibinfo{pages}{415304}
  (\bibinfo{year}{2011}{\natexlab{b}}).

\bibitem[{\citenamefont{Bru{\ss{}} et~al.}(2004)\citenamefont{Bru{\ss{}},
  D'Ariano, Lewenstein, Macchiavello, Sen(De), and Sen}}]{Bruss04}
\bibinfo{author}{\bibfnamefont{D.}~\bibnamefont{Bru{\ss{}}}},
  \bibinfo{author}{\bibfnamefont{G.}~\bibnamefont{D'Ariano}},
  \bibinfo{author}{\bibfnamefont{M.}~\bibnamefont{Lewenstein}},
  \bibinfo{author}{\bibfnamefont{C.}~\bibnamefont{Macchiavello}},
  \bibinfo{author}{\bibfnamefont{A.}~\bibnamefont{Sen(De)}}, \bibnamefont{and}
  \bibinfo{author}{\bibfnamefont{U.}~\bibnamefont{Sen}},
  \bibinfo{journal}{{Physical Review Letters}} \textbf{\bibinfo{volume}{93}},
  \bibinfo{pages}{210501} (\bibinfo{year}{2004}).

\bibitem[{\citenamefont{Zurek}(2000)}]{Zurek00}
\bibinfo{author}{\bibfnamefont{W.~H.} \bibnamefont{Zurek}},
  \bibinfo{journal}{{Annalen der Physik}} \textbf{\bibinfo{volume}{9}},
  \bibinfo{pages}{855} (\bibinfo{year}{2000}).

\end{thebibliography}

\clearpage

\setcounter{equation}{0}

\widetext
\begin{center}
{\large \bf Supplementary Material: \\ ``Quantifying the Nonclassicality of Operations''}
\end{center}

This Supplementary Material contains the proofs of the main result Theorem 1 and Properties P1-P3,  a derivation of the value of our nonclassicality measure $W$ for unitary maps, along with additional details of the calculation of $W$ for the examples given in the main text.

\section{Proof of the main result}
Theorem 1 relates to the decomposing our measure $W$ of non-classicality into two contributions. Specifically,
\begin{theorem} \label{th:RelativeEntropySumAppendix}
The quantumness of an operation $W(\M)$ is the sum of two independent contributions
\begin{equation}
W(\M) = \sup_\rho \biggl( \Srelg{\op{\Gamma} \circ \M \circ \op{\Gamma} (\rho) }{\op{\Gamma}\circ \M(\rho)} +  \Srelg{\M \circ \op{\Gamma}(\rho)}{\op{\Gamma} \circ \M \circ \op{\Gamma}(\rho)} \biggr), \label{eq:QDefinition}
\end{equation}
where the supremum is over all quantum states $\rho$.
\end{theorem}
\begin{proof}
We start with the relative entropy, defined as $\Srel{\rho}{\sigma} = \Tr\left[\rho (\log(\rho) - \log(\sigma))\right]$ \footnote{As in the main text all logarithms in this Supplementary Material are to the base 2.}, featuring under maximization in $W(\op{\M})$. Then we insert the sum of a complete set of orthonormal projectors $\sum_\alpha \op{\Pi}_\alpha =\op{1}$, where $\op{\Pi}_\alpha$ are the Kraus operators of $\op{\Gamma}$. We thus obtain
\begin{eqnarray}
\Srelg{\op{\M} \circ \op{\Gamma}(\rho)}{\op{\Gamma} \circ \op{\M}(\rho)} &=& -S\bigl(\op{\M} \circ \op{\Gamma}(\rho)\bigr) - \Tr\bigl[\op{\M} \circ \op{\Gamma}(\rho) \log(\op{\Gamma} \circ \op{\M}(\rho))\bigr] \\
 &=& -S\bigl(\op{\M} \circ \op{\Gamma}(\rho)\bigr) - \Tr\bigl[\sum_\alpha \op{\Pi}_\alpha \op{\M} \circ \op{\Gamma}(\rho) \log(\op{\Gamma} \circ \op{\M}(\rho))\bigr].
\end{eqnarray} 
where $S(\rho) = -\Tr\left[\rho \log(\rho)\right]$ is the von Neumann entropy. Next we use the fact the projective property $\left(\sum_\alpha \op{\Pi}_\alpha\right)^2 = \sum_\alpha \op{\Pi}_\alpha$ together with the cyclic property of the trace and the fact that $\op{\Pi}_\alpha$ commutes with $\op{\Gamma} \circ \op{\M}(\rho)$ and thus also with its logarithm. The above is then transformed to  
\begin{eqnarray}
\Srelg{\op{\M} \circ \op{\Gamma}(\rho)}{\op{\Gamma} \circ \op{\M}(\rho)} &=& -S\bigl(\op{\M} \circ \op{\Gamma}(\rho)\bigr) - \Tr\bigl[\sum_\alpha  \op{\Pi}_\alpha \op{\M} \circ \op{\Gamma}(\rho)\op{\Pi}_\alpha \log(\op{\Gamma} \circ \op{\M}(\rho))\bigr] \\
&=& -S\bigl(\op{\M} \circ \op{\Gamma}(\rho)\bigr) - \Tr\bigl[\op{\Gamma} \circ \op{\M} \circ \op{\Gamma}(\rho)  \log(\op{\Gamma} \circ \op{\M}(\rho))\bigr].
\end{eqnarray}
Next we add and subtract $S\bigl(\op{\Gamma} \circ \op{\M} \circ \op{\Gamma}(\rho)\bigr)$ to the righthand side to get to 
\begin{eqnarray}
\Srelg{\op{\M} \circ \op{\Gamma}(\rho)}{\op{\Gamma} \circ \op{\M}(\rho)} = S\bigl(\op{\Gamma} \circ \op{\M} \circ \op{\Gamma}(\rho)\bigr) -S\bigl(\op{\M} \circ \op{\Gamma}(\rho)\bigr)  + \Srelg{\op{\Gamma} \circ \op{\M} \circ \op{\Gamma}(\rho)}{\op{\Gamma} \circ \op{\M} (\rho)}.
\end{eqnarray}
We now expand the entropy $S\bigl(\op{\Gamma} \circ \op{\M} \circ \op{\Gamma}(\rho)\bigr)$ to give
\begin{equation}
 \Srelg{\op{\M} \circ \op{\Gamma}(\rho)}{\op{\Gamma} \circ \op{\M}(\rho)} = -\Tr[\op{\Gamma} \circ \op{\M} \circ \op{\Gamma}(\rho) \log(\op{\Gamma} \circ \op{\M} \circ \op{\Gamma}(\rho))] -S\bigl(\op{\M} \circ \op{\Gamma}(\rho)\bigr) + \Srelg{\op{\Gamma} \circ \op{\M} \circ \op{\Gamma}(\rho)}{\op{\Gamma} \circ \op{\M} (\rho)}.
\end{equation}
Now we expand $\op{\Gamma}$ and insert back the orthogonal projective operators $\op{\Pi}_\alpha$ yielding
\begin{equation}
\Srelg{\op{\M} \circ \op{\Gamma}(\rho)}{\op{\Gamma} \circ \op{\M}(\rho)} = -\Tr[\sum_\alpha \op{\Pi}_\alpha \op{\M} \circ \op{\Gamma}(\rho) \op{\Pi}_\alpha \log(\op{\Gamma} \circ \op{\M} \circ \op{\Gamma}(\rho))] -S\bigl(\op{\M} \circ \op{\Gamma}(\rho)\bigr) + \Srelg{\op{\Gamma} \circ \op{\M} \circ \op{\Gamma}(\rho)}{\op{\Gamma} \circ \op{\M} (\rho)}.
\end{equation}
Now because $\op{\Pi}_\alpha$ commutes with $\op{\Gamma} \circ \op{\M} \circ \op{\Gamma}$, we find that
\begin{equation}
 \Srelg{\op{\M} \circ \op{\Gamma}(\rho)}{\op{\Gamma} \circ \op{\M}(\rho)} = -S\bigl(\op{\M} \circ \op{\Gamma}(\rho)\bigr) -\Tr[\sum_\alpha \op{\Pi}_\alpha \op{\M} \circ \op{\Gamma}(\rho) \log(\op{\Gamma} \circ \op{\M} \circ \op{\Gamma}(\rho))] + \Srelg{\op{\Gamma} \circ \op{\M} \circ \op{\Gamma}(\rho)}{\op{\Gamma} \circ \op{\M} (\rho)}.
\end{equation}
The first two terms then form another relative entropy, leading us to
\begin{equation}
\Srelg{\op{\M} \circ \op{\Gamma}(\rho)}{\op{\Gamma} \circ \op{\M}(\rho)} = \Srelg{\op{\Gamma} \circ \op{\M} \circ \op{\Gamma}(\rho)}{\op{\Gamma} \circ \op{\M} (\rho)} + \Srelg{\op{\M} \circ \op{\Gamma} (\rho)}{\op{\Gamma} \circ \op{\M} \circ \op{\Gamma}(\rho)}.
\end{equation}
Inserting the supremum over $\rho$ then forms $W(\M)$ and completes the proof.
\end{proof}

\section{Proofs of properties P1-P3}
Here we give the proof of property P1 that the maximum in the maximization for $W$ is always attained for a pure state. 
\begin{theorem} \label{th:PureStateMax}
  Given $\sup_\rho \Srel{\M \circ \op{\Gamma}(\rho)}{\op{\Gamma} \circ \M (\rho)}$, there exists a pure state $\ket{\psi}\bra{\psi}$ such that the supremum in equation \eqref{eq:QDefinition} is attained when $\rho = \ket{\psi}\bra{\psi}$.
\end{theorem}
\begin{proof}
Imagine that we have performed maximization over only the set of pure states and found that the maximum is attained for $\ket{\psi}$. Then for some mixed state $\rho$ we can spectrally decompose it as $\rho = \sum_j \mu_j \ket{\phi_j}\bra{\phi_j}$, where $\ket{\phi_j}$ are it eigenstates. Since the relative entropy is jointly convex in its arguments \cite{NielsenChuang} this implies that
\begin{equation}
\Srel{\M \circ \op{\Gamma}(\rho)}{\op{\Gamma} \circ \M (\rho)} \leq \sum_j \mu_j \Srel{\M \circ \op{\Gamma}(\ket{\phi_j}\bra{\phi_j})}{\op{\Gamma} \circ \M (\ket{\phi_j}\bra{\phi_j})} \nonumber \leq \Srel{\M \circ \op{\Gamma}(\ket{\psi}\bra{\psi})}{\op{\Gamma} \circ \M (\ket{\psi}\bra{\psi})}.
\end{equation}
This completes the proof.
\end{proof}

Next we will consider the property P2, stating that the measure $W$ is non-increasing under the composition with classical maps.
\begin{theorem} \label{th:ClassicalComposition}
If $\op{\M}$ is some map and $W(\op{\M}_c) = 0$ then $W(\op{\M}_c \circ \op{\M}) \leq W(\op{\M})$ and $W(\op{\M} \circ \op{\M}_c) \leq W(\op{\M})$ \label{prop:CompositionIneq}.
\end{theorem}
\begin{proof}
Notice that
\begin{equation}
W(\op{\M}_c \circ \op{\M}) = \sup_\rho \Srel{\op{\M}_c \circ \op{\M} \circ \op{\Gamma}(\rho)}{\op{\M}_c \circ \op{\Gamma} \circ \op{\M}(\rho)} \leq \sup_\rho \Srel{\op{\M} \circ \op{\Gamma}(\rho)}{\op{\Gamma} \circ \op{\M}(\rho)} = W(\op{\M}) \label{eq:FirstPropertyProof1},
\end{equation}
where the last inequality is due to the monotonicity of relative entropy under completely positive operations (and thus also the strong subadditivity of the von Neumann entropy, which is equivalent to monotonicity \cite{Ruskai02}). For the reverse order
\begin{eqnarray}
W(\op{\M} \circ \op{\M}_c) &=& \sup_\rho \Srel{\op{\M} \circ \op{\Gamma} \circ \op{\M}_c(\rho)}{\op{\Gamma} \circ \op{\M} \circ \op{\M}_c(\rho)} \nonumber \\
&=& \sup_{\op{\M}_c(\rho)} \Srel{\op{\M}\circ\op{\Gamma}(\rho)}{\op{\Gamma}\circ \op{\M}(\rho)} \nonumber \\
&\leq& \sup_\rho \Srel{\op{\M}\circ\op{\Gamma}(\rho)}{\op{\Gamma}\circ \op{\M}(\rho)} = W(\op{\M}),
\end{eqnarray}
where going from second to the third line we changed the set over which we take supremum from all states to the set of states of the form $\op{\M}(\rho)$. Since this set is entirely contained in the set of all states, the inequality follows.
\end{proof}

Finally, the property P3 is proved in the following theorem.
\begin{theorem} \label{th:LOSHRMaps}
Given a set of local operations $\op{\M}_w^A, \op{\M}_w^B$ such that $W(\op{\M}_w^A \otimes \op{1}) = 0$ and $W(\op{1} \otimes \op{\M}_w^B) = 0$ then $W(\op{\M}) = 0$ for any local operation with shared randomness of the form $\op{\M} = \sum_w \gamma_w \op{\M}_w^A \otimes \op{\M}_w^B$.
\end{theorem}
\begin{proof}
Notice that since we required that $\op{\Gamma}$ be composed of local orthonormal projectors we can write it as $\op{\Gamma} = \op{\Gamma}^A\otimes\op{\Gamma}^B$ in the bipartite case. Given that $\op{\M}^A_w$ and $\op{\M}_w^B$ commute with $\op{\Gamma}^A$ and $\op{\Gamma}^B$, respectively, we also have that $\op{\M}$ commutes with $\op{\Gamma}$, establishing the result.
\end{proof}

\section{Quantumness of unitaries}
Now we turn our attention to the proof of the statement  regarding the quantumness of unitary operations.
\begin{theorem} \label{th:Unitaries}
  When $\op{\Gamma}$ acts on the entire joint Hilbert space, selecting a complete orthonormal classical basis $\ket{k}$, we have for any unitary operation $\op{U}$ that $W(\op{U}) = 0$ if and only if $\op{U} = \sum_k e^{i \phi_k} \ket{k}\bra{k} \op{P}$, where $\op{P}$ is a permutation of the classical basis and $\phi_k$ are phases. Otherwise $W(\op{U}) = \infty$. 
\end{theorem}
\begin{proof}
We proceed by computing the relative entropy $\Srelg{\op{U} \circ \op{\Gamma}(\rho)}{\op{\Gamma} \circ \op{U}(\rho)}$. First we show that if $\op{U}$ is not of the required form, then $W(\op{U}) = \infty$. Under such assumption, there exists a non-classical state $\ket{\phi_0}$ such that $\op{U} \ket{\phi_0} = \ket{j}$, where $\ket{j}$ is any classical basis state. The relative entropy $\Srel{\rho}{\sigma}$ is infinite due to the term $\Tr\left[\rho\log(\sigma)\right]$ when the kernel of $\sigma$ has a non-zero overlap with the support of $\rho$. So suppose $\ket{\phi_0} = \sum_k \alpha_k \ket{k}$ is the expansion of $\ket{\phi_0}$ in the classical basis. Then $\op{\Gamma}(\ket{\phi_0}\bra{\phi_0}) = \sum_k |\alpha_k|^2 \ket{k}\bra{k} \neq \ket{\phi_0}\bra{\phi_0}$, since $\ket{\phi_0}$ is not classical by assumption. Thus, $\op{U} \op{\Gamma}(\ket{\phi_0}\bra{\phi_0}) \op{U}^\dagger$ will in general have support across numerous classical states besides $\ket{j}\bra{j}$. However, the second argument of the relative entropy is $\op{\Gamma}(\op{U}\ket{\phi_0}\bra{\phi_0}\op{U}^\dagger) = \ket{j}\bra{j}$ and thus is a state with a kernel overlapping the support of $\op{U} \op{\Gamma}(\ket{\phi_0}\bra{\phi_0}) \op{U}^\dagger$, making the relative entropy infinite. 

Secondly, we show that if $W(\op{U}) = \infty$, then some classical state $\ket{k}$ is mapped to a non-classical state. Since the generating power is the einselected relative entropy of discord of the output state, we know that it must be bounded by $\log(d)$, where $d$ is the dimension of the joint Hilbert space. Therefore, if $W(\op{U}) = \infty$, the distinguishing power is infinite. There exists a state $\ket{\psi}$ such that
\begin{align}
  \Srelg{\op{\Gamma}( \op{U} \op{\Gamma}(\ket{\psi}\bra{\psi})\op{U}^\dagger)}{\op{\Gamma}( \op{U} \ket{\psi}\bra{\psi}\op{U}^\dagger)} = \infty.
\end{align}
  Now $\ket{\psi}$ cannot be classical, otherwise the above would vanish. So let $\ket{\psi} = \sum_k \gamma_k \ket{k}$, $\op{\Gamma}(\ket{\psi}\bra{\psi}) = \sum_k |\gamma_k|^2 \ket{k}\bra{k}$, and label the mapping of classical states as $\op{U} \ket{k} = \ket{\psi_k}$. Then we have that $\op{\Gamma}( \op{U}\op{\Gamma}(\ket{\psi}\bra{\psi})\op{U}^\dagger) = \sum_k |\gamma_k|^2 \op{\Gamma}(\ket{\psi_k}\bra{\psi_k})$ and $\op{\Gamma}( \op{U}\ket{\psi}\bra{\psi}\op{U}^\dagger) = \sum_{k,l} \gamma_k \gamma_l^* \op{\Gamma}(\ket{\psi_k}\bra{\psi_l})$. Thus we see that if for all $k$ the states $\ket{\psi_k}$ were classical, then the distinguishing power would vanish. Therefore, we must have that at least one of the states $\ket{\psi_k}$ is not classical, showing that $\op{U}$ is not of the form in the theorem statement. We have thus shown that $\op{U}$ is not a permutation matrix up to a phase if and only if $W(\op{U}) = \infty$. Conversely, when $\op{U}$ is a permutation matrix up to a phase, we know that $W$ vanishes. This completes the proof.
\end{proof}

\section{Further details for example applications}
Here we provide additional details for the examples given in the main text.

\subsection{Quantumness of a composition of operations} 
Given a sequence of operations, the quantumness $W$ is not additive under the composition of operations so that
\begin{align}
  W(\op{A}\circ\op{B}\circ\op{C}) \neq W(\op{A}) + W(\op{B}) + W(\op{C}),
\end{align}
for three operations $\op{A}$, $\op{B}$ and $\op{C}$. This is most easily demonstrated and explained by using a counter-example as described in the main text. Consider $\M = \op{H} \circ \op{\Xi}_\gamma \circ \op{H}$ in which we use the amplitude damping channel $\op{\Xi}_\gamma$ sandwiched between Hadamard gates $\op{H}$. A key observation is that the output state $\M(\rho)$, for all values of $\gamma$, has diagonal elements which are independent of the off-diagonal elements of the input state $\rho$. This immediately implies that in this case $\M$ has no distinguishing power for any input state and its quantumness comprising only of generating power. Using this and property P1 the maximization of $\Srel{\M\circ\op{\Gamma}(\rho)}{\op{\Gamma}\circ\M(\rho)}$ is therefore achieved for classical pure input states $\ket{0}$ and/or $\ket{1}$. It is straightforward to show that both are maxima. Using $\ket{0}$ we compute directly 
\begin{align}
  & \op{H} \circ \op{\Xi}_\gamma \circ \op{H}(\ket{0}\bra{0}) = \frac{1}{2}(1+\sqrt{1-\gamma}) \ket{00}\bra{00} + \frac{\gamma}{2} (\ket{10}\bra{00} + \ket{00}\bra{10}) + \frac{1}{2}(1-\sqrt{1-\gamma}) \ket{10}\bra{10} \label{eq:AmplitudeLine1} \\
  & \op{\Gamma} \circ \op{H} \circ \op{\Xi}_\gamma \circ \op{H}(\ket{0}\bra{0}) = \frac{1}{2}(1+\sqrt{1-\gamma}) \ket{00}\bra{00} + \frac{1}{2}(1-\sqrt{1-\gamma}) \ket{10}\bra{10}. \label{eq:AmplitudeLine2}
\end{align}
Here we see that in going from Eq.~\eqref{eq:AmplitudeLine1} to Eq.~\eqref{eq:AmplitudeLine2} the nonclassical off-diagonal terms are erased by the decoherence operator $\op{\Gamma}$. Overall this yields
\begin{align}
  W(\M) = \log(a)+ (a/2) \log\bigl[(1+a)/(1-a)\bigr] + (b/2) \log\bigl[ (1+b)/(1-b) \bigr],
\end{align}
where $a = \sqrt{1-\gamma}$ and $b = \sqrt{\gamma^2 - \gamma + 1}$. For $\gamma = 1$ the maximum possible quantumness of $W(\M) = 1$ is attained. At the opposite limit $\gamma = 0$, corresponding to the removal of the amplitude damping channel, we have $W(\M) = 0$, since the remaining sequence of operations $\op{H} \circ \op{1} \circ \op{H} = \op{1}$ is classical. For $0<\gamma<1$ the quantumness $W(\M)$ varies monotonically between these values. Given that $W(\op{H}) = \infty$, and hence is maximally quantum, while $W(\op{\Xi}_\gamma) = 0$ is classical, we see that $W(\op{H} \circ \op{\Xi}_\gamma \circ \op{H}) \neq W(\op{H}) + W(\op{\Xi}_\gamma) + W(\op{H})$. This non-additivity of $W$ is an expected consequence of the quantum interference between the different operations in a composition. Moreover, this example illustrates that adding a classical operation between two nonclassical operations can be used to activate quantumness.

\subsection{Discord generating map}
In the main text we consider the map $\M = \op{1} \otimes \M_B$, where $\M_B(\rho) = \op{E}_1 \rho \op{E}_1^\dagger + \op{E}_2 \rho \op{E}_2^\dagger$ and $\op{E}_1 = \ket{0}\bra{0}$, $\op{E}_2 = \ket{+}\bra{1}$. Being a local map $\M$ cannot generate entanglement, yet intriguingly when it is applied to a classical input $\sigma_c = \frac{1}{2}(\ket{0}\bra{0} \otimes \ket{0}\bra{0} + \ket{1}\bra{1} \otimes \ket{1}\bra{1})$ it can generate an output state with non-zero discord~\cite{Streltsov11}. Focusing on the non-trivial single-qubit channel $\M_B$ we again observe that the diagonal elements of $\M_B(\rho)$ are independent of the off-diagonal elements of $\rho$, leading to zero distinguishing power for any input state. As such $W(\M_B)$ is exclusively composed of generating power and its maximization is attained by a pure classical input. This is readily verified to be $\ket{1}$, as might be expected. The value of $W(\M_B)$ is then $\Srel{\M_B\circ\op{\Gamma}(\ket{1}\bra{1})}{\op{\Gamma}\circ\M_B(\ket{1}\bra{1})} = 1$ which is the maximum possible. For the extension of $\M_B$ to a two-qubit map we have $W(\op{1} \otimes \M_B)= W(\M_B)$ which is maximized by input states of the form $\ket{\psi}\otimes\ket{1}$, where $\ket{\psi}$ is an arbitrary single-qubit state.

\subsection{Depolarised rotated CNOT gate}
Next we consider a CNOT with its control rotated into the $\ket{\pm}$ basis, followed by a two-qubit depolarizing channel $\op{\Lambda}_\mu$, which gives a complete operation $\M  = \op{\Lambda}_\mu \circ (\op{H}\otimes\op{1}) \circ \op{CNOT} \circ (\op{H}\otimes\op{1})$. We find that for all values of the depolarizing probability $\mu$ the maximum generating power $W(\M\circ\op{\Gamma})$ is attained by the classical input state $\ket{0}\otimes\ket{0}$, while the maximum distinguishing power $W(\op{\Gamma}\circ\M)$ is attained by the entangled input state $\ket{\Psi} = \left(\ket{0} \otimes \ket{-} + \ket{1} \otimes \ket{+}\right)/\sqrt{2}$. As a function of $\mu$ the generating and distinguishing power are given by the relative entropies 
\begin{align}
  \Srel{\M\circ\op{\Gamma}(\ket{00}\bra{00})}{\op{\Gamma} \circ \M(\ket{00}\bra{00})} = -\frac{3}{4} \log(1-\mu) + \frac{1}{4} \log(1+3\mu), \label{eq:DistinguishingDominates}
\end{align}
and
\begin{align}
  \Srel{\M \circ \op{\Gamma}(\ket{\Psi}\bra{\Psi})}{\op{\Gamma} \circ \M(\ket{\Psi}\bra{\Psi})} = \frac{3}{4}(1-\mu) \log(1-\mu) + \frac{1+3\mu}{4} \log(1+3\mu), \label{eq:GeneratingDominates}
\end{align}
respectively. Performing the maximization for $W(\M)$ reveals that there exists a certain threshold $\mu_c$, such that whenever $\mu<\mu_c$ we have $W(\M) = W(\M\circ\op{\Gamma})$, while for $\mu>\mu_c$ we have $W(\M) = W(\op{\Gamma}\circ\M)$. The maximum therefore switches between the generating and distinguishing power at the transition point $\mu = \mu_c$, which is found by direct substitution to be $\mu_c = 2/3$, where both Eqs.~\eqref{eq:GeneratingDominates} and \eqref{eq:DistinguishingDominates} evaluate to $\log(3)/2$. Note that unlike the previous examples, where $W(\M)$ being composed purely of generating power coincides with a vanishing distinguishing power, here both the generating power and distinguishing power are non-zero for all $0<\mu\leq 1$. 

\end{document}